\documentclass[conference,a4paper]{IEEEtran}

\usepackage[utf8]{inputenc} 
\usepackage[T1]{fontenc}
\usepackage{url}
\usepackage{ifthen}
\usepackage{cite}
\usepackage[cmex10]{amsmath}

\usepackage{amssymb,amsfonts,enumerate,cleveref}
\usepackage{mathdots}
\newtheorem{theorem}{Theorem}
\newtheorem{lemma}{Lemma}
\newtheorem{conjecture}{Conjecture}
\newtheorem{proposition}{Proposition}
\newenvironment{proof}{\paragraph*{Proof}}{\hfill$\square$}

\begin{document}
\title{Further Progress on the GM-MDS Conjecture for Reed-Solomon Codes} 

\author{
  \IEEEauthorblockN{Hikmet Yildiz and Babak Hassibi}
  \IEEEauthorblockA{Department of Electrical Engineering\\
                    California Institute of Technology\\ 
                    Pasadena, California 91125\\
                    Email: \{hyildiz, hassibi\}@caltech.edu}
}

\maketitle

%%%%%%%%%%%%%%%%%%%%%%%%%%%%%%%%%%%%
%
%     ABSTRACT
%
%%%%%%%%%%%%%%%%%%%%%%%%%%%%%%%%%%%%

\begin{abstract}
  Designing good error correcting codes whose generator matrix has a support constraint, i.e., one for which only certain entries of the generator matrix are allowed to be non-zero, has found many recent applications, including in distributed coding and storage, multiple access networks, and weakly secure data exchange. The dual problem, where the parity check matrix has a support constraint, comes up in the design of locally repairable codes. The central problem here is to design codes with the largest possible minimum distance, subject to the given support constraint on the generator matrix. An upper bound on the minimum distance can be obtained through a set of singleton bounds, which can be alternatively thought of as a cut-set bound. Furthermore, it is well known that, if the field size is large enough, any random generator matrix obeying the support constraint will achieve the maximum minimum distance with high probability. Since random codes are not easy to decode, structured codes with efficient decoders, e.g., Reed-Solomon codes, are much more desirable. The GM-MDS conjecture of Dau et al states that the maximum minimum distance over all codes satisfying the generator matrix support constraint can be obtained by a Reed Solomon code. If true, this would have significant consequences. The conjecture has been proven for several special case: when the dimension of the code k is less than or equal to five, when the number of distinct support sets on the rows of the generator matrix m, say, is less than or equal to three, or when the generator matrix is sparsest and balanced. In this paper, we report on further progress on the GM-MDS conjecture. In particular, we show that the conjecture is true for all m less than equal to six. This generalizes all previous known results (except for the sparsest and balanced case, which is a very special support constraint).
\end{abstract}

%%%%%%%%%%%%%%%%%%%%%%%%%%%%%%%%%%%%
%
%     INTRODUCTION
%
%%%%%%%%%%%%%%%%%%%%%%%%%%%%%%%%%%%%

\section{Introduction}

There has been a recent interest in finding an MDS code with a generator matrix constrained on the support. This problem appears in many areas such as distributed coding and storage, multiple access networks, where each relay nodes has access to a subset of the sources \cite{halbawi2014distributed}, \cite{dau2015simple}, and weakly secure data exchange, where users have a subset of the data packets and want to exchange them without revealing information to eavesdroppers \cite{yan2013algorithms}, \cite{yan2014weakly}.

For a linear code with length $n$ and dimension $k$, the singleton bound on the minimum distance is $d_{\min}\leq n-k+1$. To achieve this bound, namely for MDS codes, any $k$ columns of the generator matrix $G$ should be linearly independent. Let $S_i$ be the set of positions of the zeros in the $i$th row of $G$. Then, for any subset $I\subset[k]$, the columns indexed in $\bigcap_{i\in I}S_i$ will have zeros in all their entries in $I$.\footnote{$[n]$ represents the set $\{1,2,\dots,n\}$}
Since those columns need to be linearly independent, for any nonempty $I\subset[k]$,
\begin{equation}
	\label{condition_intro}
	k-|I|\geq \left|\bigcap_{i\in I}S_i\right|
\end{equation}
is a necessary condition for the code to be MDS.

It is not hard to show that generating a random matrix with a constrained support satisfying (\ref{condition_intro}) will result in an MDS code with high probability if the field size is large enough \cite{ho2008byzantine}. Nonetheless, since random codes are not easy to decode, it is more preferable to design structured codes like Reed-Solomon codes, which have efficient decoders. The GM-MDS conjecture stated by Dau \textit{et al.} \cite{dau2014existence} describes (\ref{condition_intro}) as also a sufficient condition for the existence of a Reed-Solomon code whose generator matrix satisfies the support constraints.

Although the conjecture has many equivalent versions and partial proofs have been proposed in \cite{halbawi2014distributed},\cite{yan2014weakly},\cite{dau2014existence},\cite{heidarzadeh2017algebraic},
it has not been proven yet in general.
Heidarzadeh \textit{et al.} \cite{heidarzadeh2017algebraic}
proved it for $k\leq 5$.
Halbawi \textit{et al.} \cite{halbawi2014distributed}
proved the statement for $m\leq 3$ if there are $m$ distinct support sets on the rows of the generator matrix.
In \cite{halbawi2016balanced}, \cite{halbawi2016balanced2}, the result is proven when the generator matrix is sparsest and balanced.
Yan \textit{et al.} \cite{yan2014weakly}
give a partial induction step, a way to reduce the problem from $k$ to $k-1$ if one of the inequalities in (\ref{condition_intro}) holds with equality for some $I$ such that $|I|=k-1$.

We should mention that there is a related problem where, given a support constraint on the generator matrix, one would like to find a code with the largest minimum distance. This is because not every support constraint will admit an MDS code. Again it can be shown that, for a large enough field size, a random generator matrix satisfying the support constraints achieves the maximum minimum distance with high probability. In \cite{dau2014existence}, it has been shown that the existence of Reed-Solomon codes that achieve the maximum minimum distance is equivalent to the GM-MDS conjecture studied in this paper. Finally, we should mention that the dual problem where the support constraint is on the parity check matrix, is of interest in locally-repairable codes \cite{tamo2016bounds}, \cite{yekhanin2012locally}.

In this paper, we will group the rows with the same support constraint and prove the GM-MDS conjecture for ${m\leq 6}$, which improves all the previous results except the sparsest and balanced case.
Furthermore, we will give a more extensive way for reducing the problem, which covers any equality case, not limited to ${|I|=k-1}$.

The rest of the paper is organized as follows. In the next section we formulate the problem and, in fact, introduce a slightly more general problem including multisets which will be of use in our proof. The proof of the main result appears in Section 3 and the paper concludes with Section 4.

%%%%%%%%%%%%%%%%%%%%%%%%%%%%%%%%%%%%
%
%     PROBLEM SETUP
%
%%%%%%%%%%%%%%%%%%%%%%%%%%%%%%%%%%%%

\section{Problem Setup}

%%%%%%%%%%%%%%%%%%%%%%%%%%%
% GM-MDS conjecture
%%%%%%%%%%%%%%%%%%%%%%%%%%%
\subsection{GM-MDS conjecture}
Consider the generator matrix of a generalized Reed-Solomon code:
\begin{equation}
G_{RS}=\begin{bmatrix}
1 & 1 & \dots & 1\\
\alpha_1 & \alpha_2 & \dots & \alpha_n\\
\vdots & \vdots &  & \vdots \\
\alpha_1^{k-1} & \alpha_2^{k-1} & \dots & \alpha_n^{k-1}
\end{bmatrix}
\end{equation}
for distinct $\alpha_1,\dots, \alpha_n\in\mathbb{F}_q$. For a nonsingular  $T\in\mathbb{F}_q^{k\times k}$, define the $k\times n$ generator matrix
\begin{equation}
G=T\,G_{RS}
\end{equation}
The GM-MDS conjecture \cite{dau2014existence} is
\begin{conjecture}
	\label{conj_original}
	If $S_1,\dots, S_k\subset[n]$ satisfy for any nonempty $I\subset[k]$,
	\begin{equation}
	\label{condition_original}
	k-|I|\geq \left|\bigcap_{i\in I}S_i\right|,
	\end{equation}
	then, there exists $q,\alpha_1,\dots,\alpha_n\in\mathbb{F}_q$ and a nonsingular ${T\in\mathbb{F}_q^{k\times k}}$ such that $G_{ij}=0$ for all $i\in[k]$ and $j\in S_i$ where $G=TG_{RS}$.
	\hfill$\diamond$
\end{conjecture}

%%%%%%%%%%%%%%%%%%%%%%%%%%%
% Grouping equal subsets
%%%%%%%%%%%%%%%%%%%%%%%%%%%
\subsection{Grouping  equal subsets}
We will group the equal subsets and represent the position of zeros in the rows of the generator matrix by the sets $S_1,\dots,S_m$ with multiplicities $r_1,\dots,r_m$ where $\sum_{i=1}^mr_i=k$.
That is, the first $r_1$ rows of $G$ will have zeros at positions in $S_1$, the next $r_2$ rows will have zeros at positions in $S_2$ and so on.
Then, the condition (\ref{condition_original}) on these sets becomes
\begin{equation}
	\label{condition_becomes}
	k-\sum_{i\in I}r_i\geq \left|\bigcap_{i\in I}S_i\right|
\end{equation}
for any nonempty $I\subset[m]$.

Dau \textit{et al.} \cite{dau2014existence} have shown that it is sufficient to prove only the case when $|S_i|=k-1$ for $i\in[k]$ in \Cref{conj_original}. We will first show that in the grouped case, it is sufficient to prove the conjecture when $|S_i|=k-r_i$ for $i\in[m]$:

Suppose $(S_i)_{i=1}^m$ satisfies condition (\ref{condition_becomes}). Let ${S'_i=S_i\cup S''_i}$ for ${i\in[m]}$ and $S''_1,\dots,S''_m$ be any partition of ${[n']\backslash[n]}$ such that ${|S''_i|=k-r_i-|S_i|}$ where ${n'=n+\sum_{i=1}^mk-r_i-|S_i|}$.
Note that ${k-r_i-|S_i|\geq 0}$ due to (\ref{condition_becomes}) for $I=\{i\}$.
Then, $(S'_i)_{i=1}^m$ will also satisfy the condition (\ref{condition_becomes}).
Define ${G'_{RS}\in\mathbb{F}^{k\times n'}}$ similarly by introducing new variables $\alpha_{n+1},\dots,\alpha_{n'}$.
If $q,\alpha_1,\dots,\alpha_{n'},T$ is a solution for $(S'_i)_{i=1}^m$, then $q,\alpha_1,\dots,\alpha_{n},T$ will be a solution for $(S_i)_{i=1}^m$. 

After this assumption, the span of the grouped rows in $T$ will be uniquely determined. Therefore, we can analyze the singularity of one example:
\begin{equation}\arraycolsep=2.6pt\small
\label{matrixT}
T=\left[\begin{array}{cccccc}
0 & & \prod_{j\in S_1}\!-\alpha_j & \dots & \sum_{j\in S_1}\!-\alpha_j &  1\\
 & \iddots & & \iddots &\iddots \\
\prod_{j\in S_1}\!-\alpha_j & \dots & \sum_{j\in S_1}\!-\alpha_j & 1& & 0\\
\hline
&&\vdots\\
\hline
0 & & \prod_{j\in S_m}\!-\alpha_j & \dots & \sum_{j\in S_m}\!-\alpha_j& 1\\
& \iddots & & \iddots & \iddots \\
\prod_{j\in S_m}\!-\alpha_j & \dots & \sum_{j\in S_m}\!-\alpha_j & 1& & 0
\end{array}\right]
\end{equation}
which is partitioned into $m$ blocks where the $i$th block has ${k-|S_i|}$ rows formed by the coefficients of the polynomials $p_i$, $xp_i$, $\dots$, $x^{k-1-|S_i|}p_i$ where $p_i(x)=\prod_{j\in S_i}(x-\alpha_j)$.
The precondition
\begin{equation}
	\sum_{i=1}^mk-|S_i|=\sum_{i=1}^mr_i=k
	\iff \sum_{i=1}^m |S_i|=(m-1)k
\end{equation}
ensures that the matrix $T$ is $k\times k$ square.
Furthermore, in the multiplication $T\:G_{RS}$, the rows will consist of the substitution of $\alpha_i$'s in the polynomials of the form $x^{\ell}p_j$, which will have zeros at the desired positions.

As a result, we end up with an equivalent conjecture to \Cref{conj_original}:\\

\begin{conjecture}
	\label{conj_main}
	For $m\geq 2$, let $S_1,\dots, S_m\subset[n]$ such that $|S_i|\leq k-1$, $\sum_{i=1}^m |S_i|=(m-1)k$ and for any nonempty $I\subset[m]$,
	\begin{equation}
		\label{condition}
		k-\left|\bigcap_{i\in I}S_i\right|\geq \sum_{i\in I}k-|S_i|.
	\end{equation}
	Then, $\det T$ (which is a multivariate polynomial of $\alpha_i$'s) is not identically zero,
	where $T$ is given by (\ref{matrixT}).
	\hfill$\diamond$\\
\end{conjecture}

From now on, we will assume that $\alpha_1,\dots,\alpha_n$ are indeterminates and write $\det T=0$ or $\det T\neq 0$ to indicate that the determinant is identically zero or nonzero, respectively.

%%%%%%%%%%%%%%%%%%%%%%%%%%%
% Extension to multisets
%%%%%%%%%%%%%%%%%%%%%%%%%%%
\subsection{Extension to multisets}
Multisets are the generalization of the sets where multiple instances of the set elements are allowed \cite{blizard1988multiset}. The multiset extension for the sets $S_1,\dots,S_m$ will be useful later in the proof of our main results. Although multisets have no meaning regarding the positions of the zeros in the generator matrix, we can still define the matrix $T$ in (\ref{matrixT}) for $S_i$'s being multisets, in which case, the polynomials $p_i(x)=\prod_{j\in S_i}(x-\alpha_j)$ may have multiple roots. We will not write the conjecture for multisets; however, the fact that the condition (\ref{condition}) is necessary can be extended for multisets as well:\\

\begin{theorem}
	\label{thm_necessary}
	Let $S_1,\dots,S_m$ be multisets in the universe $[n]$ such that $|S_i|\leq k-1$, $\sum_{i=1}^m |S_i|=(m-1)k$. Define the matrix $T$ as in (\ref{matrixT}). If $\det T\neq 0$, then the condition (\ref{condition}) holds for any nonempty $I\subset[m]$.
	\hfill$\diamond$\\
\end{theorem}

Now, we will introduce \Cref{prop_polyeq} and \Cref{lemma_main} regarding the multiset extension, which will be handy later when proving our main results. \Cref{prop_polyeq} is straightforward by definition of $T$ in (\ref{matrixT}).
\\

\begin{proposition}
	\label{prop_polyeq}
	Let $S_1,\dots,S_m$ be multisets in the universe $[n]$ such that $|S_i|\leq k-1$ and ${\sum_{i=1}^m |S_i|=(m-1)k}$. Let $p_1,\dots,p_m$ be the polynomials defined as ${p_i(x)=\prod_{j\in S_i}(x-\alpha_j)}$ for all $i\in[m]$.
	Then, $\det T=0$ if and only if there exists some polynomials $q_1,\dots,q_m$, not all zero, such that
	$\deg q_i \leq k-1-\deg p_i$ and
	$\sum_{i=1}^mq_ip_i=0$.
	\hfill$\diamond$\\
\end{proposition}

\begin{lemma}
	\label{lemma_multisetreduce}
	Let $S_1,\dots,S_m$ be multisets in the universe $[n]$ such that $|S_i|\leq k-1$, $\sum_{i=1}^m |S_i|=(m-1)k$ and $\bigcap_{i=1}^{m}S_i=\emptyset$. Define $S_0=\bigcap_{i=1}^{m-1}S_i$ and $S'_i=S_i\backslash S_0$ for $i\in[m]$. Let $T$ and $T'$ be defined as in (\ref{matrixT}) for $(S_i)_{i=1}^m$ and $(S'_i)_{i=1}^m$ respectively (for $(S'_i)_{i=1}^m$, we use $k'=k-|S_0|$).
	Then,
	\begin{enumerate}
		\item $\det T'\neq 0$ implies $\det T\neq 0$.
		\item $(S_i)_{i=1}^m$ satisfies the condition (\ref{condition}) if and only if $(S'_i)_{i=1}^m$ satisfies the condition (\ref{condition}) for $k'=k-|S_0|$.
		\hfill$\diamond$
	\end{enumerate}
\end{lemma}

\textit{Sketch of Proof:}
For the first part, if $\det T=0$, using \Cref{prop_polyeq}, we have $\sum_{i=1}^mq_ip_i=0$, which yields ${p_0\triangleq\gcd_{i\in[m-1]}p_i}$ divides $q_mp_m$. Since $S_0\cap S_m=\emptyset$, $(p_0,p_m)=1$ and $p_0$ divides $q_m$. Then, dividing all the terms in $\sum_{i=1}^mq_ip_i$ by $p_0$ and using \Cref{prop_polyeq} again completes the proof.
The second part is straightforward by definition of $(S'_i)_{i=1}^m$.
\hfill$\square$

%%%%%%%%%%%%%%%%%%%%%%%%%%%%%%%%%%%%
%
%     MAIN RESULTS
%
%%%%%%%%%%%%%%%%%%%%%%%%%%%%%%%%%%%%

\section{Main Results}
Due to the lack of a complete proof for \Cref{conj_main}, we will apply the minimal counterexample method in order to present all our findings and to show that the conjecture holds for all $m\leq 6$. If \Cref{conj_main} is not true, then, there will be at least one counterexample which satisfies the conditions in \Cref{conj_main} but for which, $\det T=0$.  Among these counterexamples, there will be one (or many) that is minimal with regards to the parameters $(m,n,k)$ when considered in lexicographical order. In \Cref{lemma_main}, some necessary conditions are listed for a minimal counterexample. Note that these conditions are not necessary for any counterexample but for a \emph{minimal} counterexample. 

It will be needed in the statement of \Cref{lemma_main} to define a new collection of sets $(Q_i)_{i=1}^n$ where $Q_i=\{t: i\in S_t\}\subset[m]$.
\\

\begin{lemma}
	\label{lemma_main}
	If \Cref{conj_main} is not true and $(S_i)_{i=1}^m$ is a counterexample such that $(m,n,k)$ is the smallest possible in the lexicographical order\footnote{It turns out that \Cref{lemma_main} is also true for different orderings of $(m,n,k)$.}, then, the following\footnote{Although there are eight conditions listed, the last five are consequences of the first three.} must be true:
	\begin{enumerate}[i.]
		\item For any nonempty $I\subset [m]$ such that $|I|\neq 1,m$,
		\begin{equation}
			k-1-\left|\bigcap_{i\in I}S_i\right| \geq \sum_{i\in I} k-|S_i|
		\end{equation}
		
		\item For any $i\neq j\in [n]$,
		\begin{equation}
			Q_i\cap Q_j=\emptyset \implies Q_j=[m]\backslash Q_i
		\end{equation}

		\item For any $i, j\in [n]$,
		\begin{equation}
			|Q_i\cup Q_j|\neq m-1
		\end{equation}
		
		\item For any $i\neq j\in[m]$, there exists $\ell\in[n]$ such that $i\in Q_{\ell}$ and $j\notin Q_{\ell}$ (Equivalently, $S_i\not\subset S_j$).

		\item For any $i\in[n]$, $|Q_i|\leq m-3$.
		
		\item For any $i\in [n]$,
		\begin{equation}
			|Q_i|\geq \frac{n-1}{k-1}
		\end{equation}
		Furthermore, since $n\geq k+1$, $|Q_i|\geq 2$.

		\item There exists $i\in[n]$ such that $|Q_i|\geq 3$.
		
		\item If $|Q_i|=2$ for some $i\in[n]$, then, for any $j\in[n]$, $|Q_i\cap Q_j|\geq 1$.
		\hfill$\diamond$\\
	\end{enumerate}
\end{lemma}
\begin{proof}
	\begin{enumerate}[i.]
		\item % i
		Since $\det T=0$, by \Cref{prop_polyeq}, there exist polynomials $q_1,\dots,q_m$, not all zero, such that ${\deg q_i\leq k-1-\deg p_i}$ and ${\sum_{i=1}^mq_ip_i=0}$.
		Assume the contrary. Then, there exists some $I$ with ${2\leq |I|\leq m-1}$ such that
		\begin{equation}
			\label{eq_proof1}
			k-\left|\bigcap_{i\in I}S_i\right| = \sum_{i\in I} k-|S_i|
		\end{equation}
		Let $J=\{0\}\cup[m]-I$, $S_0=\bigcap_{i\in I}S_i$, $p_0=\gcd_{i\in I}p_i$ and $S'_i=S_i-S_0$, $p'_i=p_i/p_0$ for $i\in I$.
		Then, by (\ref{eq_proof1}), $(S'_i)_{i\in I}$ and $(S_i)_{i\in J}$ satisfy the conditions in \Cref{conj_main} (for $(S'_i)_{i\in I}$, we use $k'=k-|S_0|$).
		By the minimality of $(S_i)_{i=1}^m$, \Cref{conj_main} is true for both $(S'_i)_{i\in I}$ and $(S_i)_{i\in J}$.
		We can write that
		\begin{equation}
			0=\sum_{i=1}^m q_ip_i=\left(\sum_{i\in I}q_ip'_i\right)p_0+\sum_{i\in [m]-I}q_ip_i
		\end{equation}
		Using \Cref{prop_polyeq} for $(S_i)_{i\in J}$, we get $q_i=0$ for $i\notin I$ and $\sum_{i\in I}q_ip'_i=0$. Then, by using \Cref{prop_polyeq} for $(S'_i)_{i\in I}$, we get $q_i=0$ for $i\in I$. Contradiction.
		\\
		\item % ii
		Assume the contrary. Hence, there exists $i\neq j\in[n]$ such that $Q_i\cap Q_j=\emptyset$ and $Q_i\cup Q_j\neq [m]$. W.l.o.g. assume that $Q_{n-1}\cup Q_{n}=\emptyset$ and $m\notin Q_{n-1}\cup Q_{n}$. For all $i\in[m]$, define the sets
		\begin{equation}
			S'_i=\begin{cases}
			S_i & n\notin S_i \: (i\notin Q_{n})\\
			(S_i\backslash\{n\})\cup\{n-1\} & n\in S_i \: (i\in Q_{n})
			\end{cases}
		\end{equation}
		Since $Q_{n-1}\cap Q_{n}=\emptyset$, if $n\in S_i$, then $n-1\notin S_i$ yielding $|S'_i|=|S_i|$. Hence, we can define $T'$ for $S'_i$'s and
		\begin{equation}
			\det T'=\left.\det T\right\vert_{\alpha_{n}=\alpha_{n-1}}=0
		\end{equation}
		Since $[m]\not\subset Q_{n-1}\cup Q_{n}$, $\left|\bigcap_{i=1}^mS'_i\right|= \left|\bigcap_{i=1}^mS_i\right|=0$.
		Since $[n-2]\cap S_i=[n-2]\cap S'_i$, for any nonempty $I\subset[m]$,
		\begin{IEEEeqnarray}{rl}
			\left|\bigcap_{i\in I}S'_i\right|-\left|\bigcap_{i\in I}S_i\right|
			&= \left|\{n-1\}\cap\bigcap_{i\in I}S'_i\right| \nonumber\\
			&\quad\:-\left|\{n,n-1\}\cap\bigcap_{i\in I}S_i\right| \\
			&\leq 1
		\end{IEEEeqnarray}
		Therefore, for $|I|\neq 1,m$,
		\begin{equation}
			k-\left|\bigcap_{i\in I}S'_i\right|\geq k-1-\left|\bigcap_{i\in I}S_i\right|\geq \sum_{i\in I} k-|S'_i|
		\end{equation}
		Hence, $(S'_i)_{i=1}^m$ is also a counterexample with parameters $(m,n-1,k)$. Contradiction.
		\\
		\item % iii
		Assume the contrary. W.l.o.g. assume that ${Q_{n-1}\cup Q_{n}=[m-1]}$.
		For all $i\in[m]$, define the multisets
		\begin{equation}
			S'_i=\begin{cases}
			S_i & n\notin S_i \: (i\notin Q_{n})\\
			(S_i\backslash\{n\})\uplus\{n-1\} & n\in S_i \: (i\in Q_{n})
			\end{cases}
		\end{equation}
		where $\uplus$ is the multiset summation.
		Then, similar to (ii), $\det T'=0$. We have that ${|\bigcap_{i=1}^mS'_i|=|\bigcap_{i=1}^mS_i|=0}$.
		Denote by $\mu_S(j)$ the multiplicity of $j$ in $S$. Then, for any $i\in[m]$,
		$\mu_{S'_i}(n)=0$,
		$\mu_{S'_i}(n-1)=\mu_{S_i}(n-1)+\mu_{S_i}(n)$, and
		$\mu_{S'_i}(j)=\mu_{S_i}(j)$ for $j\in[n-2]$.
		Then, for any nonempty $I\subset[m]$,
		
		\begin{IEEEeqnarray}{rl}
			\left|\bigcap_{i\in I}S'_i\right|-\left|\bigcap_{i\in I}S_i\right|
			&= \sum_{j=1}^{n-1}\min_{i\in I}\mu_{S'_i}(j)
			\!-\!\! \sum_{j=1}^n\min_{i\in I}\mu_{S_i}(j)\\
			&\leq \min_{i\in I}\mu_{S'_i}(n\!-\!1)
			\!-\! \min_{i\in I}\mu_{S_i}(n\!-\!1)\quad\:\:\:\\
			&\leq \min_{i\in I}\mu_{S_i}(n-1) +1 \nonumber\\
			&\quad\quad-\min_{i\in I}\mu_{S_i}(n-1)\\
			&=1
		\end{IEEEeqnarray}
		
		Therefore, for $|I|\neq 1,m$,
		\begin{equation}
			k-\left|\bigcap_{i\in I}S'_i\right|\geq k-1-\left|\bigcap_{i\in I}S_i\right|\geq \sum_{i\in I} k-|S'_i|
		\end{equation}
		So, $(S'_i)_{i=1}^m$ satisfies the conditions in \Cref{conj_main} except they are multisets.
		However, we can apply \Cref{lemma_multisetreduce} by defining ${S'_0\triangleq\bigcap_{i=1}^{m-1}S'_i=\{n-1\}}$ and $S''_i=S'_i\backslash\{n-1\}$.
		Note that $S''_i$'s are normal sets i.e. they do not contain any element with multiplicity more than one. Then, by \Cref{lemma_multisetreduce}, $(S''_i)_{i=1}^m$ is also a counterexample with parameters $(m,n-1,k-1)$. Contradiction.
		\\
		\item % iv
		Assume the contrary. Then, there exists $i\neq j$ such that $S_i\subset S_j$. By (i), 
		\begin{equation}
			k-1-|S_i\cap S_j|\geq 2k-|S_i|-|S_j|
		\end{equation}
		yielding $|S_j|\geq k+1$. Contradiction.
		\\
		\item % v
		Assume that $Q_1=[m-1]$. Then, $S_0=\bigcap_{i=1}^{m-1}S_i\neq\emptyset$. Apply \Cref{lemma_multisetreduce}. Hence, $(S'_i)_{i=1}^m$ is also a counter example with smaller parameters. Contradiction.
		
		Assume that $Q_1=[m-2]$. Then, by (iii), for any ${j\in[n]}$, ${|Q_1\cup Q_j|\neq m-1}$. Then, either $|Q_1\cup Q_j|=m-2$ meaning $Q_j\subset [m-2]$ or $|Q_1\cup Q_j|=m$ meaning ${m-1,m\in Q_j}$. Hence, there is no $Q_j$ containing $m-1$ but not $m$. Contradiction due to (iv).
		\\
		\item % vi
		Assume that $Q_1=[\ell]$. Then, by (ii), for ${j=2,\dots, n}$ either $[\ell]\cap Q_j\neq\emptyset$ or $Q_j=[m]\backslash [\ell]$. Hence, we can partition the set $\{2,\dots,n\}$ into
		\begin{align}
			J_1&=\{2\leq j\leq n : [\ell]\cap Q_j\neq\emptyset \},\\
			J_2&=\{2\leq j\leq n : Q_j=[m]\backslash [\ell] \}
		\end{align}
		If $j\in J_1$, then $Q_j\cap [\ell]\neq\emptyset$, which implies ${j\in\bigcup_{i=1}^{\ell} (S_i\backslash\{1\})}$. Hence, $J_1\subset\bigcup_{i=1}^{\ell} (S_i\backslash\{1\})$. Therefore,
		\begin{equation}\small
			\label{eq_J1}
			|J_1|\leq \left|\bigcup_{i=1}^{\ell} (S_i\backslash\{1\})\right|\leq \sum_{i=1}^{\ell} |S_i\backslash\{1\}|=-\ell+\sum_{i=1}^{\ell} |S_i|
		\end{equation}
		where we use the fact that $1\in S_i$ for $i\in [\ell]$ since ${Q_1=[\ell]}$.
		
		If $j\in J_2$, then $Q_j=[m]\backslash [\ell]$, which implies ${j\in\bigcap_{i=\ell+1}^mS_i}$. Hence, $J_2\subset\bigcap_{i=\ell+1}^m S_i$. So,
		\begin{equation}
		\label{eq_J2}
		k-|J_2|\geq k-\left|\bigcap_{i=\ell+1}^m S_i\right|\geq \sum_{i=\ell+1}^mk-|S_i|
		\end{equation}
		As a result of (\ref{eq_J1}) and (\ref{eq_J2}),
		\begin{align}
			n-1&=|J_1|+|J_2|\\
			&\leq -\ell+k-(m-\ell)k+\sum_{i=1}^m |S_i|\\
			&=\ell(k-1)
		\end{align}
		Thus, $\ell\geq \frac{n-1}{k-1}$.
		
		By (i), $k-1-|S_1\cap S_2|\geq 2k-|S_1|-|S_2|$. Hence, $n=|\bigcup_{i=1}^mS_i|\geq |S_1\cup S_2|\geq k+1$.
		\\
		\item % vii
		Assume the contrary. Then, for all $i\in[n]$, $|Q_i|=2$. By (v) and (vi), $m\geq 5$.
		Then, by (ii), for any $i,j\in[n]$, $Q_i\cap Q_j\neq\emptyset$; so, either $Q_i=Q_j$ or $|Q_i\cap Q_j|=1$.
		W.l.o.g. assume  that $Q_1=\{1,2\}$. Then, by (iv), for any $3\leq j\leq m$, there exists $\ell_j\in[n]$ such that $j\in Q_{\ell_j}$ and $2\notin Q_{\ell_j}$. Then, $Q_{\ell_j}=\{1,j\}$. Then, for any $i\in[n]$, $1\in Q_i$. Then, there is no $Q_i$ containing $2$ but not $1$. Contradiction due to (iv).
		\\
		\item % viii
		Corollary of (ii) and (v).
		
	\end{enumerate}
\end{proof}

%%%%%%%%%%%%%%%%%%%%%%%%%%%
% Consequences of Lemma 2
%%%%%%%%%%%%%%%%%%%%%%%%%%%
\subsection{Consequences of \Cref{lemma_main}}
Firstly, \Cref{lemma_main} allows us to make the assumptions listed from (i) to (viii) when proving \Cref{conj_main}. If \Cref{conj_main} is true under these assumptions, then it must be also true without these assumptions; otherwise, it will lead to a contradiction for the minimal counterexample. For example, the condition (i) implies that ``\emph{it is enough to prove \Cref{conj_main} only for the case where all the inequalities in (\ref{condition}) are strict}''.

Secondly, if one of the conditions listed in \Cref{lemma_main} does not hold, then, the problem can be reduced to the one with a smaller parameter $m,n$ or $k$. The way in which it is reduced can be found in the proof of \Cref{lemma_main} for parts (i)-(iii).

Thirdly, it can help us to solve the problem for small parameters. The conditions (v) and (vii) already imply that \Cref{conj_main} is true for $m\leq 5$ because by condition (vii), there exists a set $Q_i$ with size at least $3$, whose size is upper bounded by $m-3$ in condition (v). By a little more work, we can also solve $m=6$ as shown in \Cref{thm_uptosix}:\\

\begin{theorem}
	\label{thm_uptosix}
	\Cref{conj_main} is true for $m\leq 6$.
	\hfill$\diamond$\\
\end{theorem}
\begin{proof}
	Assume the contrary. Then, there exists a minimal counterexample $S_1,\dots,S_m$ such that $m\leq 6$.
	By (v) and (vii), $m\geq 6$. Hence, $m=6$.
	By (v) and (vi), $|Q_i|\in\{2,3\}$ for all $i\in[n]$. 
	If $Q_i$ and $Q_j$ are size $3$, then $|Q_i\cap Q_j|\neq 1$ by (iii). There will be three cases:
	
	\textit{Case 1.} $|\{Q_i: i\in[n], |Q_i|=2 \}|\geq 2$.
	Assume that $Q_1$ and $Q_2$ are two different sets of size $2$. By (viii), their intersection must have exactly one element. W.l.o.g. assume that ${Q_1=\{1,2\}}$ and $Q_2=\{1,3\}$. By (viii), for any $i\in[n]$, $|Q_i\cap Q_1|\geq 1$ and $|Q_i\cap Q_2|\geq 1$; hence, either $1\in Q_i$ or $2,3\in Q_i$. For any $j=4,5,6$, there is a set containing $j$ but not $1$, which has to be $\{j,2,3\}$. Let $Q_3=\{2,3,4\}, {Q_4=\{2,3,5\}}, Q_5=\{2,3,6\}$. Let $Q_6$ be the set containing $4$ but not $2$. Then, $1\in Q_6$. Since $|Q_4\cup Q_6|,|Q_5\cup Q_6|\neq 5$, we have $5,6\in Q_6$, which means $|Q_6|\geq 4$. Contradiction.
	
	\textit{Case 2.} $|\{Q_i: i\in[n], |Q_i|=2 \}|=1$.
	W.l.o.g. let ${Q_1=\{1,2\}}$. If $Q_i\neq Q_1$, then $|Q_i|=3$ and by (viii) either $1\in Q_i$ or $2\in Q_i$. Let $Q_2$ be the set containing $3$ but not $2$. Then, wlog. $Q_2=\{1,3,4\}$. Let $Q_3$ be the set containing $3$ but not $1$. Then, $Q_3=\{2,3,4\}$. Let $Q_4$ be the set containing $5$ but not $2$. Then, $Q_4=\{1,5,x\}$ for some $x$. Since $|Q_2\cap Q_4|\neq 1$, $x\in\{3,4\}$. Then, $|Q_3\cap Q_4|=1$. Contradiction.
	
	\textit{Case 3.} For all $i\in[n]$, $|Q_i|=3$. Let $Q_1=\{1,2,3\}$. Let $Q_2$ contain $2$ but not $3$. Then, wlog $Q_2=\{1,2,4\}$. Let $Q_3$ contain $2$ but not $1$. Then, $Q_3=\{2,3,4\}$. Let $Q_4$ contain $3$ but not $2$. Then, $Q_4=\{1,3,4\}$. Let $5\in Q_5$. Then, $Q_5$ has at least one element from $\{1,2,3,4\}$. Wlog assume that $1\in Q_5$. Then, $Q_5=\{1,2,5\}$. $|Q_3\cap Q_5|=1$. Contradiction.
\end{proof}

%%%%%%%%%%%%%%%%%%%%%%%%%%%%%%%%%%%%
%
%     CONCLUSION
%
%%%%%%%%%%%%%%%%%%%%%%%%%%%%%%%%%%%%

\section{Conclusion}
We have established the correctness of the GM-MDS conjecture of Dau \textit{et al.} for $m\leq 6$, where $m$ is the number of distinct support sets defined on the rows of the generator matrix.
The result subsumes all earlier known results on the GM-MDS conjecture except for those pertaining to sparsest and balanced generator matrices. Our results followed has a careful study of properties that must hold for any minimal counterexample to the conjecture.
It remains to be seen whether this approach can be extended to prove the conjecture for values of $m$ beyond $6$.

%%%%%%%%%%%%%%%%%%%%%%%%%%%%%%%%%%%%
%
%     REFERENCES
%
%%%%%%%%%%%%%%%%%%%%%%%%%%%%%%%%%%%%
\bibliographystyle{IEEEtran}
\bibliography{bibliofile}

\end{document}